\documentclass[conference]{IEEEtran}
\IEEEoverridecommandlockouts

\pdfoutput=1

\usepackage{amsfonts}
\usepackage{amsmath}
\usepackage{amsthm}
\usepackage{amssymb}
\usepackage{mathtools}
\usepackage{commath}
\usepackage{mathrsfs}
\usepackage{mathdots}

\usepackage[sort,comma, numbers]{natbib}

\usepackage{microtype}              

\usepackage{hyperref}
\hypersetup{colorlinks=false}
\usepackage{cleveref}

\usepackage{siunitx}

\newtheorem{proposition}{Proposition}

\usepackage{booktabs}

\def\BibTeX{{\rm B\kern-.05em{\sc i\kern-.025em b}\kern-.08em
    T\kern-.1667em\lower.7ex\hbox{E}\kern-.125emX}}
\begin{document}

\title{%
    Numerically robust Gaussian state estimation \\ 
    with singular observation noise
    \thanks{%
        NK was supported by a research grant (42062) from VILLUM FONDEN, partly funded by the Novo Nordisk Foundation through the Center for Basic Machine Learning Research in Life Science (NNF20OC0062606), and received funding from the European Research Council (ERC) under the European Union's Horizon programme (grant agreement 101125993).
        FT was partially supported by the Wallenberg AI, Autonomous Systems and Software Program (WASP) funded by the Knut and Alice Wallenberg Foundation. 

    }
}

\author{\IEEEauthorblockN{Nicholas Krämer}
\IEEEauthorblockA{\textit{Department of Applied Mathematics and Computer Science} \\
\textit{Technical University of Denmark}\\
Kongens Lyngby, Denmark \\
pekra@dtu.dk}
\and
\IEEEauthorblockN{Filip Tronarp}
\IEEEauthorblockA{\textit{Centre for Mathematical Sciences} \\
\textit{Lund University}\\
Lund, Sweden \\
filip.tronarp@matstat.lu.se}
}

\maketitle

\begin{abstract}
This article proposes numerically robust algorithms for Gaussian state estimation with singular observation noise.
Our approach combines a series of basis changes with Bayes' rule,
transforming the singular estimation problem into a nonsingular one with reduced state dimension.
In addition to ensuring low runtime and numerical stability,
our proposal facilitates marginal-likelihood computations and Gauss--Markov representations of the posterior process.
We analyse the proposed method's computational savings and numerical robustness and validate our findings in a series of simulations.
\end{abstract}


\section{Introduction}
\label{sec:introduction}
Let $\{u_t\}_{t=0}^T \subseteq \mathbb{R}^n$ and $\{w_t\}_{t=0}^T \subseteq \mathbb{R}^r$ be pairwise independent Gaussian variables and consider the state-space model ($x_{-1}\coloneqq 0$)
\begin{subequations} \label{equation-full-ssm}
\begin{align}
x_t &= \Phi_t x_{t-1} + Q_t u_t, & t&=0, ..., T\\ 
y_t &= C_t x_t + F_t w_t, & t&=0, ..., T
\end{align}
\end{subequations}
with transition parameters $Q_t, \Phi_t \in \mathbb{R}^{n \times n}$, and matrices $C_t \in \mathbb{R}^{m \times n}$ and $F_t \in \mathbb{R}^{(\ell + r) \times r}$
The goals are to infer the state sequence $\{x_t\}_{t=0}^T \subseteq \mathbb{R}^n$ from the observed sequence $\{y_t\}_{t=0}^T \subseteq \mathbb{R}^{\ell + r}$, and to evaluate the marginal likelihood of the observations.
Both problems can be solved by, for instance,
the Kalman filter \citep{Kalman1961} and Rauch--Tung--Striebel smoother \citep{Rauch1965}.
However, when the observation model is singular, which means $\ell > 0$, then parts of the state $\{x_t\}_{t=0}^T$ are fully determined in a subspace of $\mathbb{R}^n$.
Consequently, there are computational savings on the table if the state estimation problem can be reduced to the orthogonal complement of said subspace.

More specifically, we study the state estimation problem under the following assumptions \citep{AitElFquih2011,tse1970optimal,Anderson1979}:
First, assume \mbox{$\ell + r \leq n$} and that the matrices $\{C_t\}_t$ have full rank; otherwise, the system would include redundant observations.
Second, assume the matrices $\{F_t\}_t$ and $\{Q_t\}_t$ have full rank -- however, the rank $r$ of $F_t$ could be zero.

Existing work has approached this problem as follows.
\citet{tse1970optimal} develop a method for state estimation in singular observation models based on linear observer theory but do not consider smoothing.
\citet{ko2007state} examine the problem from the perspective of Kalman filtering theory. However, their discussion is confined to the case of time-invariance.
\citet[Section 11.3]{Anderson1979} and \citet{Ghanbarpourasl2022} discuss the same setting.
\citet{AitElFquih2011} develop smoothing algorithms based on model reduction via singular value decompositions.
None of these works discusses numerically robust algorithms, which replace (numerically fragile) covariance-arithmetic with QR decompositions; more on this later.
This gap was partly filled by \citet{Psiaki1999}, who develops numerically robust filtering and smoothing recursions -- however, only for information form parametrisations.
None of the mentioned articles discuss marginal likelihoods.
\citet{Geng2022} develop a general smoothing algorithm via the batch formulation of the estimation problem but do not identify the subspaces in which the state is fully determined and thus leave dimensionality reduction with corresponding runtime improvements on the table.
For a brief survey on constrained state estimation, refer to \citet{simon2010kalman}.

\subsection{Contributions and outline}
\Cref{table:summary_of_contributions} summarises how our work fills all the abovementioned gaps.
\begin{table}[t]
\caption{Contributions. Information form discussion: \Cref{proposition-backward-kernels}}
\label{table:summary_of_contributions}
\begin{center}
\begin{tabular}{| l | c | c | c | c | c |}
\hline
    & Ours & \citep{AitElFquih2011} & \citep{Psiaki1999} & \citep{tse1970optimal,Ghanbarpourasl2022} & \citep{ko2007state}\\ 
\hline 
Filtering 
    & + & - & + & + & +
    \\
Smoothing 
    & + & + & + & - & -
    \\
Marginal likelihood 
    & + & - & - & - & -
    \\
Numerically robust
    & + & - & + & - & -
    \\
Time-varying models 
    & + & + & + & + & -
    \\
Covariance form
    & + & + & - & + & +
    \\
Information form
    & $\sim$ & - & + & - & -
    \\
\hline 
\end{tabular}
\end{center}
\end{table}
Furthermore, and unlike existing algorithms, ours separates into an ``offline'' (\Cref{sec:model_reduction}) and an ``online'' part (\Cref{sec:gaussian_state_estimation}), which means that significant parts of the computation can be performed before the first observation is encountered. 
Numerical experiments (\Cref{sec:experiments}) corroborate our algorithm's efficiency.
On a side note,
\Cref{sec:gaussian_state_estimation} argues for a formulation of Gaussian state estimation that, while not new \citep{kramer2025numerically,tronarp2022fenrir}, might be underappreciated in the literature.

\subsection{Notation}
Denote the transpose of a matrix $M$ by $M^*$ and
abbreviate $x_{0:t} \coloneqq \{x_0, ..., x_t\}$.
Let $I_k$ be the identity matrix with $k \in \mathbb{N}$ rows and columns. 
$\mathcal{O}(\cdot)$ is the usual ``big-Oh'' complexity.

The QR decomposition factorises $M \in \mathbb{R}^{n \times m}$, $n \geq m$ into the product of an orthogonal matrix $Q$ and an upper triangular matrix $R$: $M = QR$. 
We distinguish the \emph{complete QR decomposition} ($Q$ square, $R$ shaped like $M$) and the \emph{thin QR decomposition} ($Q$ shaped like $M$, $R$ square). 
Similarly, the \emph{complete and thin LQ decompositions} factorise a matrix into the product of a lower-triangular and an orthogonal matrix. We implement them by applying QR decompositions to the transpose of a matrix, followed by transposing the results.
The \emph{complete and thin QL decompositions} decompose $M$ into the product of an orthogonal matrix and a lower triangular matrix. Both can be implemented with QR decompositions: For any $k \in \mathbb{N}$, let $F_k \in \mathbb{R}^{k \times k}$ be a matrix with ones on the antidiagonal and zeros elsewhere. Then, a QR decomposition of $M F_n$ yields
\begin{align}
M &= M F_m F_m = QR F_m = (QF_n) (F_nRF_m).
\end{align}
Due to the structure of $F_k$, the matrix $Q F_n$ is orthogonal, and $F_n R F_m$ is lower triangular; thus, a QL decomposition.

\section{Model reduction (before seeing observations)}
\label{sec:model_reduction}
The essence of implementing state estimation algorithms on models with $n$ state, $r$ noisy, and $\ell$ noise-free dimensions is a reduction to a state-space model with fewer dimensions. 
\Cref{sub:one_step_model} explains this for a one-step model, and \Cref{sub:time_varying_model} applies the one-step result to the system in \Cref{equation-full-ssm}.

\subsection{One-step model}
\label{sub:one_step_model}
Let $\Phi$, $Q$, $C$, and $F$ be versions of $\Phi_t$, $Q_t$, $C_t$, and $F_t$ without time indices.
For pairwise independent random variables $u, z \in \mathbb{R}^n$, and $w \in \mathbb{R}^{r}$, define a state $x \in \mathbb{R}^n$ and an observation $y \in \mathbb{R}^{\ell + r}$ via
\begin{align} \label{equation-to-be-reduced-model}
x = \Phi z + Qu, \quad y = C x + F w.
\end{align}
The model in \Cref{equation-to-be-reduced-model} provides the building blocks for reducing the model in \Cref{equation-full-ssm}. Therefore, reducing  \Cref{equation-to-be-reduced-model} is the main ingredient for rewriting a singular state-estimation task as a lower-dimensional, nonsingular one.
We use a sequence of QL and LQ decompositions to achieve this reduction, which aligns the singular and nonsingular components in the observation $y$ with the corresponding components in the state $x$.

First, a complete QL decomposition of $F$ yields column-orthogonal matrices  $V_\mathsf{u} \in \mathbb{R}^{\ell +r \times r}$ and $V_\mathsf{c} \in \mathbb{R}^{\ell +r \times \ell }$, and a lower-triangular matrix $L_\mathsf{u} \in \mathbb{R}^{r \times r}$ that satisfy
\begin{align}
F = \begin{pmatrix}V_\mathsf{u} & V_\mathsf{c} \end{pmatrix} \begin{pmatrix} L_\mathsf{u} \\ 0 \end{pmatrix}.
\end{align}
The reason for using a QL decomposition instead of a QR decomposition is that if $w$ is a standard Gaussian variable, $L_\mathsf{u}$ is the Cholesky factor of $V_\mathsf{u}^* F w$ if we use QL.

Second, a complete LQ decomposition of $V_\mathsf{c}^* C$ yields column-orthogonal matrices  $W_\mathsf{c} \in \mathbb{R}^{n \times \ell}$ and $W_\mathsf{u} \in \mathbb{R}^{n \times (n-\ell)}$, and a lower-triangular matrix $S_\mathsf{c} \in \mathbb{R}^{\ell \times \ell}$ that satisfy
\begin{align}
V_\mathsf{c}^* C  = \begin{pmatrix} S_\mathsf{c} & 0 \end{pmatrix} \begin{pmatrix} W_\mathsf{c}^* \\ W_\mathsf{u}^*\end{pmatrix}.
\end{align}
Together, these two factorisations realign the components in \Cref{equation-to-be-reduced-model} as follows. 
Introduce
\begin{align} \label{equation-state-transformation}
x^\mathsf{c} \coloneqq  W_\mathsf{c}^*x,
\quad
x^\mathsf{u} \coloneqq  W_\mathsf{u}^*x, 
\quad
y^\mathsf{c} \coloneqq V_\mathsf{c}^*y,
\quad
y^\mathsf{u} \coloneqq V_\mathsf{u}^*y.
\end{align}
The superscripts in $x^\mathsf{c}$, $x^\mathsf{u}$, $y^\mathsf{c}$, and $y^\mathsf{u}$ indicate ``constrained'' and ``unconstrained'' parts of $x$ and $y$. 
\Cref{equation-to-be-reduced-model} becomes
\begin{subequations} \label{equation-state-space-rearranged}
\begin{align}
\begin{pmatrix} x^\mathsf{c} \\ x^\mathsf{u} \end{pmatrix}
&=
\begin{pmatrix} W_\mathsf{c}^* \\ W_\mathsf{u}^*\end{pmatrix}
\Phi z
+
\begin{pmatrix} W_\mathsf{c}^* \\ W_\mathsf{u}^* \end{pmatrix}
Q u 
\label{equation-state-space-rearranged-x}
\\ 
\begin{pmatrix} y^\mathsf{u} \\ y^\mathsf{c} \end{pmatrix}
&=
\begin{pmatrix} V_\mathsf{u}^* C W_\mathsf{c} & V_\mathsf{u}^* C W_\mathsf{u} \\ S_\mathsf{c} & 0 \end{pmatrix} 
\begin{pmatrix} x^\mathsf{c} \\ x^\mathsf{u} \end{pmatrix}
+
\begin{pmatrix} L_\mathsf{u}\\ 0 \end{pmatrix} w.
\end{align}
\end{subequations}
\Cref{equation-state-space-rearranged} shows how $y^\mathsf{c}$ uniquely determines $x^\mathsf{u}$ (and vice versa), since $C$ has full rank, which means $S_\mathsf{c}$ is invertible.

Third, removing the dependence of $x^\mathsf{c}$ on $x^\mathsf{u}$ will allow eliminating $x^\mathsf{c}$ from the model, thereby reducing the dimensionality. 
A complete LQ decomposition yields lower-triangular $Z_\mathsf{c} \in \mathbb{R}^{\ell\times \ell}$, $Z_\mathsf{u} \in \mathbb{R}^{(n-\ell) \times (n-\ell)}$, a dense $Z_\star \in \mathbb{R}^{(n-\ell)\times \ell}$, and column-orthogonal $U_\mathsf{c} \in \mathbb{R}^{n \times n-\ell}$ and $U_\mathsf{u} \in \mathbb{R}^{n \times \ell}$ such that
\begin{align}
\begin{pmatrix}
W_\mathsf{c}^* Q \\ 
W_\mathsf{u}^* Q
\end{pmatrix}
=
\begin{pmatrix}
Z_\mathsf{c} & 0 \\ Z_\star & Z_\mathsf{u}
\end{pmatrix}
\begin{pmatrix}
U_\mathsf{c}^* \\ U_\mathsf{u}^*
\end{pmatrix}
\end{align}
holds.
Let $G \coloneqq L_\star (L_1)^{-1} \in \mathbb{R}^{(n-\ell) \times \ell}$ and observe
\begin{align}\label{equation-secret-conditioning}
\begin{pmatrix}
Z_\mathsf{c} & 0 \\ Z_\star & Z_\mathsf{u}
\end{pmatrix}
=
\begin{pmatrix}
I_{\ell} & 0 \\ G & I_{n-\ell}
\end{pmatrix}
\begin{pmatrix}
Z_\mathsf{c} & 0 \\ 0 & Z_\mathsf{u}
\end{pmatrix}.
\end{align}
\Cref{equation-secret-conditioning} implements Bayes' rule because it factorises $p(x^\mathsf{c}, x^\mathsf{u})$ into  $p(x^\mathsf{u} \,|\, x^\mathsf{u}) p(x^\mathsf{u})$; compare \Cref{equation-secret-conditioning} to \Cref{proposition-backward-kernels} in \Cref{sec:gaussian_state_estimation}.
This sequence of two QR-style decompositions followed by Bayes' rule is central to our work.

Left-multiply \Cref{equation-state-space-rearranged-x} with the inverse of the left term on the right-hand side of \Cref{equation-secret-conditioning}, abbreviate 
\begin{align} \label{equation-transformed-noise}
u^\mathsf{c} \coloneqq U_\mathsf{c}^* u, \quad u^\mathsf{u} \coloneqq U_\mathsf{u}^* u,
\end{align} and sort the terms in the resulting expression to obtain
\begin{subequations}
\begin{align}
x^\mathsf{c} &= W_\mathsf{c}^* \Phi z + Z_\mathsf{c} u^\mathsf{c}
\\
x^\mathsf{u} &= (W_\mathsf{u}^* - G W_\mathsf{c}^*) \Phi z + G x^\mathsf{c} + Z_\mathsf{u} u^\mathsf{u}
\\ 
y^\mathsf{u} &= V_\mathsf{u}^* C W_\mathsf{c} x^\mathsf{c} + V_\mathsf{u}^* C W_\mathsf{u} x^\mathsf{u} + L_\mathsf{u} w
\\ 
y^\mathsf{c} &= S_\mathsf{c} x^\mathsf{c}.
\end{align}
\end{subequations}
Finally, eliminate $x^\mathsf{c}$ from this system via $x^\mathsf{c} = (S_c)^{-1} y^\mathsf{c}$:
\begin{subequations} \label{equation-reduced-model}
\begin{align}
y^\mathsf{c} &= S_\mathsf{c} W_\mathsf{c}^* \Phi z + S_\mathsf{c} Z_\mathsf{c} u^\mathsf{c}
\\
x^\mathsf{u} &= (W_\mathsf{u}^* - G W_\mathsf{c}^*) \Phi z + G (S_\mathsf{c})^{-1} y^\mathsf{c} + Z_\mathsf{u} u^\mathsf{u}
\\ 
y^\mathsf{u} &= V_\mathsf{u}^* C W_\mathsf{c} (S_\mathsf{c})^{-1} y^\mathsf{c} + V_\mathsf{u}^* C W_\mathsf{u} x^\mathsf{u} + L_\mathsf{u} w.
\end{align}
\end{subequations}
\Cref{equation-reduced-model} is a reduced version of \Cref{equation-to-be-reduced-model}.

\subsection{Time-varying model}
\label{sub:time_varying_model}
We return to the state-estimation task from \Cref{equation-full-ssm} and assume that the process noises $u_{0:T}$ and the observation noises $w_{0:T}$ are pairwise independent Gaussian variables with zero mean and unit covariance. However, any mean and covariance would apply.
Let $x^\mathsf{c}_t$, $x^\mathsf{u}_t$, $y^\mathsf{c}_t$, $y^\mathsf{u}_t$, $u^\mathsf{c}_t$, and $u^\mathsf{u}_t$ be transformations of $x_t$, $y_t$, and $u_t$ according to \Cref{equation-state-transformation,equation-transformed-noise}. Define corresponding matrices $V_{\mathsf{c}, t}$, $V_{\mathsf{u}, t}$, $W_{\mathsf{c}, t}$, $W_{\mathsf{u}, t}$, $S_{\mathsf{c}, t}$, $Z_{\mathsf{c}, t}$, $G_t$, $Z_{\mathsf{u}, t}$, and $L_{\mathsf{u}, t}$ that are produced like in \Cref{sec:model_reduction}.
The only difference is the additional time index.

At any $t$, $x_t$ can be reconstructed from $x_t^\mathsf{u}$ via
\begin{align} \label{equation-reconstruct-x}
x_t = W_{\mathsf{u}, t} x^\mathsf{u}_t + W_{\mathsf{c}, t} (S_\mathsf{c})^{-1} y_t^\mathsf{c}.
\end{align}
Applying the results from \Cref{sec:model_reduction} to \Cref{equation-full-ssm}, setting $z = x_t$, and plugging in \Cref{equation-reconstruct-x}, the reduced state-estimation task becomes the following:
The unconstrained state components $x^\mathsf{u}_t$ transition like
\begin{subequations} \label{equation-reduced-state-transition-t}
\begin{align}
x^\mathsf{u}_0 =&\,  G_0 S_{\mathsf{c}, 0}^{-1} y^\mathsf{c}_0 + Z_{\mathsf{u}, 0} u^\mathsf{u}_0
\label{equation-reduced-state-transition-t-0}
\\ 
x^\mathsf{u}_t =&\, 
\Psi_{1, t} x^\mathsf{u}_{t-1} + \Psi_{2, t} y_{t-1}^\mathsf{c} + G_t S_{\mathsf{c}, t}^{-1} y^\mathsf{c}_t + Z_{\mathsf{u}, t} u^\mathsf{u}_t
\label{equation-reduced-state-transition-t-t}
\\ 
\Psi_{1, t} \coloneqq&\, (W_{\mathsf{u}, t}^* - G_t W_{\mathsf{c}, t}) \Phi_t W_{\mathsf{u}, t-1}
\\ 
\Psi_{2, t} \coloneqq&\, (W_{\mathsf{u}, t}^* - G_t W_{\mathsf{c}, t}) \Phi_t W_{\mathsf{c}, t-1} S_{\mathsf{c}, t-1}^{-1}
\end{align}
\end{subequations}
where $1 \leq t \leq T$ holds.
The observation $y^\mathsf{u}_t$ becomes
\begin{align} \label{equation-reduced-unconstrained-observation-t}
y^\mathsf{u}_t &= V_{\mathsf{u}, t}^* C_t W_{\mathsf{c}, t}^* (S_{\mathsf{c}, t})^{-1} y^\mathsf{c}_t + V_{\mathsf{u}, t}^* C_t W_{\mathsf{u}, t}^* x_t^\mathsf{u} + L_{\mathsf{u}, t} w_t 
\end{align}
for $t=0, ..., T$. 
The constraint $y^\mathsf{c}_t$ relates to the others like 
\begin{subequations} \label{equation-reduced-constrained-observation-t}
\begin{align}
y^\mathsf{c}_0 =&\, S_{\mathsf{c}, 0} Z_{\mathsf{c}, 0} u^\mathsf{c}_0
\label{equation-reduced-constrained-observation-t-0}
\\
y^\mathsf{c}_t =&\,\Lambda_{1,t} x^\mathsf{u}_{t-1} + \Lambda_{2,t} y_{t-1}^\mathsf{c}  + S_{\mathsf{c}, t} Z_{\mathsf{c}, t} u^\mathsf{c}_t 
\label{equation-reduced-constrained-observation-t-t}
\\ 
\Lambda_{1,t} \coloneqq&\,  S_{\mathsf{c}, t} W_{\mathsf{c}, t}^* \Phi_t W_{\mathsf{u}, t-1}
\\ 
\Lambda_{2,t} \coloneqq&\, S_{\mathsf{c}, t} W_{\mathsf{c}, t}^* \Phi_t W_{\mathsf{c}, t-1} (S_{\mathsf{c}, t-1})^{-1} 
\end{align}
\end{subequations}
for $t=1, ..., T$.
The difference between the above equations and \Cref{equation-reduced-model} involves additional time indices and the application of \Cref{equation-reconstruct-x}.

\section{State estimation (after seeing observations)}
\label{sec:gaussian_state_estimation}
Since $U_{\mathsf{c}, t}$ and $U_{\mathsf{c}, t}$ span pairwise orthogonal spaces, and since $u_{0:T}$ and $w_{0:T}$ are pairwise independent, $u^\mathsf{c}_{0:T}$, $u^\mathsf{u}_{0:T}$, and $w_{0:T}$ must be pairwise independent. 
This independence implies linear-time state-estimation algorithms because the posterior distribution admits a sequential factorisation: \Cref{proposition:posterior-distribution-factorises}.

\begin{proposition}\label{proposition:posterior-distribution-factorises}
The posterior distribution over the unconstrained state $p(x^\mathsf{u}_{0:T} \mid y^\mathsf{u}_{0:T}, y^\mathsf{c}_{0:T})$ equals
\begin{align}\label{equation-posterior-distribution-factorises}
p(x^\mathsf{u}_{T} \mid y^\mathsf{u}_{0:T}, y^\mathsf{c}_{0:T})
\prod_{t=1}^T p(x^\mathsf{u}_{t-1}\mid x^\mathsf{u}_{t}, y^\mathsf{u}_{0:t-1}, y^\mathsf{c}_{0:t}).
\end{align}
\end{proposition} 
\begin{proof}
Factorise $p(x^\mathsf{u}_{0:T} \mid y^\mathsf{u}_{0:T}, y^\mathsf{c}_{0:T})$ into the equivalent
\begin{align}
p(x^\mathsf{u}_{T} \mid y^\mathsf{u}_{0:T}, y^\mathsf{c}_{0:T})
\prod_{t=1}^T p(x^\mathsf{u}_{t-1} \mid x^\mathsf{u}_t, y^\mathsf{u}_{0:T}, y^\mathsf{c}_{0:T})
\end{align}
and apply
\begin{align}
\begin{split}
&p(x^\mathsf{u}_{t-1} \mid x^\mathsf{u}_t, y^\mathsf{u}_{0:T}, y^\mathsf{c}_{0:T}) =
p(x^\mathsf{u}_{t-1} \mid x^\mathsf{u}_t, y^\mathsf{u}_{0:t-1}, y^\mathsf{c}_{0:t})
\end{split}
\end{align}
which holds since all process and observation noises are pairwise independent.
\end{proof}
\Cref{sec:the_algorithm} will explain numerically robust, linear-time algorithms for evaluating each term in \Cref{equation-posterior-distribution-factorises}.
These terms include the filtering distributions $\{p(x^\mathsf{u}_{t} \mid y^\mathsf{u}_{0:t}, y^\mathsf{c}_{0:t})\}_{t=0}^T$, thus implementing algorithms based on \Cref{proposition:posterior-distribution-factorises} yields a numerically robust Kalman filter in the reduced system.
\Cref{equation-posterior-distribution-factorises} can also be used to parametrise the smoothing distributions $\{p(x^\mathsf{u}_t \mid y^\mathsf{u}_{0:T}, y^\mathsf{c}_{0:T}\}_{t=0}^T$, via a backwards-sequence of marginalisations, numerically robust algorithms for which are well-known \citep{grewal2014kalman}. As such, \Cref{proposition:posterior-distribution-factorises} also implies a numerically robust Rauch--Tung-Striebel smoother in the reduced system.
More on filtering and smoothing in \Cref{sec:the_algorithm}.

A statement similar to \Cref{proposition:posterior-distribution-factorises} can be made about the marginal likelihood of the observations:
\begin{proposition}\label{proposition:marginal-likelihood-factorises}
The marginal likelihood of the observations $p(y^\mathsf{u}_{0:T}, y^\mathsf{c}_{0:T})$ factorises into the equivalent (read $y_{0:-1} \coloneqq \emptyset$) 
\begin{align}\label{equation-marginal-likelihood-factorisation}
\prod_{t=0}^T p(y^\mathsf{c}_t \mid y^\mathsf{c}_{0:t-1}, y^\mathsf{u}_{0:t-1}) p(y^\mathsf{u}_t \mid y^\mathsf{c}_{0:t}, y^\mathsf{u}_{0:t-1}).
\end{align}
\end{proposition}
\begin{proof}
Apply $p(y^\mathsf{u}, y^\mathsf{c}) = p(y^\mathsf{u} \mid y^\mathsf{c}) p(y^\mathsf{c})$ liberally.
\end{proof}
Each term in \Cref{equation-marginal-likelihood-factorisation} emerges during the forward pass of computing the terms in \Cref{proposition:posterior-distribution-factorises}; more in \Cref{sec:the_algorithm}.

\subsection{Numerically robust Gaussian conditioning}
Next, we zoom in on implementing sequential algorithms based on \Cref{proposition:posterior-distribution-factorises,proposition:marginal-likelihood-factorises}.
Numerically robust state estimation in this context hinges on the following routines.
\begin{proposition}\label{proposition-backward-kernels}
Let $x$ and $y$ be Gaussian variables relating as
\begin{align}
x \sim \mathcal{N}(m, L L^*), \quad y \mid x \sim \mathcal{N}(A x + b, B B^*).
\end{align}
Then, parametrisations of $p(y)$ and $p(x \mid y)$ can be computed from $m$, $L$, $A$, $b$, and $B$, without ever forming $L L^*$ or $B B^*$.
\end{proposition}
\begin{proof}
The following algorithm proves \Cref{proposition-backward-kernels}.
Complete LQ decompose the Cholesky factor of the joint law $p(y, x)$,
\begin{align}
\begin{pmatrix}
A L & B \\ 
L & 0
\end{pmatrix}
=
\begin{pmatrix}
L_1 & 0 \\ 
L_\star & L_2
\end{pmatrix}
\begin{pmatrix}
T_1 \\ T_2
\end{pmatrix}
\end{align}
Let $K \coloneqq L_\star L_1^{-1}$. 
The parameters of $p(y)$ and $p(x \mid y)$ are 
\begin{subequations}
\begin{align}
y &\sim \mathcal{N}(A m + b, L_1 L_1^*), \\
x \mid y &\sim \mathcal{N}(m - K (A m + b - y), L_2 L_2^*), \label{equation-proof-backward-kernel}
\end{align}
\end{subequations}
because the identity
\begin{align}
\begin{pmatrix}
A L & B \\ 
L & 0
\end{pmatrix}
\begin{pmatrix}
A L & B \\ 
L & 0
\end{pmatrix}^*
&=
\begin{pmatrix}
L_1 L_1^* & L_1 L_\star^* \\ 
L_\star L_1^* & L_\star L_\star^* + L_2 L_2^*
\end{pmatrix}
\end{align}
holds: $L_1 L_1^*$ matches the known formula for Gaussian marginals, and $L_2 L_2^*$ does the same for conditional covariances,
\begin{align}
L_2 L_2^* = L L^* - L_\star L_\star^* = L L^* - K L_1 L_1^* K^*
\end{align}
which proves the claim.
\end{proof}
The fact that \Cref{proposition-backward-kernels} never forms a full covariance matrix in combination with the numerical stability of QR decompositions makes algorithms that condition Gaussians via \Cref{proposition-backward-kernels} numerically robust.
To derive a version of our algorithms that avoids QR or exclusively uses precision matrices, replace \Cref{proposition-backward-kernels} accordingly; see also \citet{Psiaki1999}'s information form algorithm.

\Cref{proposition-backward-kernels} does not involve realisations of observed variables; changing that turns \Cref{proposition-backward-kernels} into a numerically robust implementation of Bayes' rule for Gaussian variables:
\begin{proposition}\label{proposition:bayesian-update}
Consider the setting of \Cref{proposition-backward-kernels} and denote by $\hat y$ a realisation of $y$. Then, evaluate $p(x \mid y=\hat y)$ and $p(y=\hat y)$ by computing $p(x \mid y)$ and $p(y)$ via \Cref{proposition-backward-kernels}, followed by evaluating the conditional and marginal at $y = \hat y$.
\end{proposition}
The reason why we separate \Cref{proposition-backward-kernels} from \Cref{proposition:bayesian-update} is that \Cref{equation-posterior-distribution-factorises} involves the backward transitions $p(x_{t-1} \mid x_t)$, parametrisations of which require \Cref{proposition-backward-kernels}, not \Cref{proposition:bayesian-update}.
This distinction means that \Cref{proposition-backward-kernels} is a critical part of numerically robust Gaussian smoothers \citep{gibson2005robust}.

\subsection{The algorithm}
\label{sec:the_algorithm}

Denote the realisations of $y_t^\mathsf{c}$ and $y_t^\mathsf{u}$ by $\hat y_t^\mathsf{c}$ and $\hat y_t^\mathsf{u}$, respectively. The distinction between an observed variable and its realisation is important to illustrate in which sense \Cref{proposition:bayesian-update} is applicable.
For example, $p(x^\mathsf{u}_t \mid x_{t-1}^\mathsf{u}, y_{t-1}^\mathsf{c}, y_t^\mathsf{c})$ is not a transition density from $x_{t-1}^\mathsf{u}$ to $x_t^\mathsf{u}$, but $p(x^\mathsf{u}_t \mid x_{t-1}^\mathsf{u}, \hat y_{t-1}^\mathsf{c}, \hat y_t^\mathsf{c})$ is since the observed variables $\hat y_{t-1}^\mathsf{c}$ and $\hat y_t^\mathsf{c}$ are realised.

When reading the following algorithm, recall that all transition densities and likelihoods are Gaussian and parametrised as in \Cref{sec:gaussian_state_estimation}. Therefore, all numerical linear algebra reduces to \Cref{proposition:bayesian-update,proposition-backward-kernels}.

To initialise, evaluate $p(\hat y_0^\mathsf{c})$ (\Cref{equation-reduced-constrained-observation-t-0}) and store the result.
Plug the prior  $p(x_0^\mathsf{u} \mid \hat y_0^\mathsf{c})$  (\Cref{equation-reduced-state-transition-t-0}) and the likelihood $p(\hat y_0^\mathsf{u} \mid x_0^\mathsf{u}, \hat y_0^\mathsf{c})$ (\Cref{equation-reduced-unconstrained-observation-t}) into \Cref{proposition:bayesian-update}. This parametrises the filtering distribution $p(x_0^\mathsf{u} \mid \hat y^\mathsf{u}_{0}, \hat y^\mathsf{c}_0)$ and the likelihood increment $p(\hat y_0^\mathsf{u} \mid \hat y_0^\mathsf{c})$. Store both.

Then, for each $t=1, ..., T$, do the following. Assume that at each stage $t$, a parametrisation of the filtering distribution $p(x_{t-1}^\mathsf{u} \mid \hat y^\mathsf{u}_{0:t-1}, \hat y^\mathsf{c}_{0:t-1})$ is available.
\begin{enumerate}
    \item Plug the filtering distribution  $p(x_{t-1}^\mathsf{u} \mid \hat y^\mathsf{u}_{0:t-1}, \hat y^\mathsf{c}_{t-1})$ and the likelihood $p(\hat y_t^\mathsf{c} \mid x_{0:t-1}^\mathsf{u}, \hat y_{t-1}^\mathsf{c})$ (\Cref{equation-reduced-constrained-observation-t-t}) into \Cref{proposition:bayesian-update}. 
    This procedure parametrises the conditional distribution $p(x_{t-1}^\mathsf{u} \mid \hat y^\mathsf{u}_{0:t-1}, \hat y^\mathsf{c}_{0:t})$ and the likelihood increment $p(\hat y_t^\mathsf{c} \mid \hat y_{0:t-1}^\mathsf{c}, \hat y_{0:t-1}^\mathsf{u})$. Store both.
    \item Plug the conditional distribution $p(x_{t-1}^\mathsf{u} \mid \hat y^\mathsf{u}_{0:t-1}, \hat y^\mathsf{c}_{0:t})$ and the transition density $p(x^\mathsf{u}_t \mid x_{t-1}^\mathsf{u}, \hat y_{t-1}^\mathsf{c}, \hat y_t^\mathsf{c})$ (\Cref{equation-reduced-state-transition-t-t}) into \Cref{proposition-backward-kernels}. This parametrises the predicted distribution $p(x_{t}^\mathsf{u} \mid \hat y^\mathsf{u}_{0:t-1}, \hat y^\mathsf{c}_{0:t})$ and the backward density $p(x_{t-1}^\mathsf{u} \mid x_t^\mathsf{u}, \hat y^\mathsf{u}_{0:t-1}, \hat y^\mathsf{c}_{0:t})$. 
    For a Gaussian filter, do not store the backward density; otherwise, do.
    \item Plug the predicted distribution $p(x_{t}^\mathsf{u} \mid \hat y^\mathsf{u}_{0:t-1}, \hat y^\mathsf{c}_{0:t})$ and the likelihood $p(\hat y_t^\mathsf{u} \mid x_t^\mathsf{u}, \hat y_t^\mathsf{c})$ (\Cref{equation-reduced-unconstrained-observation-t}) into \Cref{proposition:bayesian-update}. This parametrises the next filtering distribution $p(x_{t}^\mathsf{u} \mid \hat y^\mathsf{u}_{0:t}, \hat y^\mathsf{c}_{0:t})$ and the likelihood increment $p(\hat y^\mathsf{u}_t \mid \hat y^\mathsf{u}_{0:t-1}, \hat y^\mathsf{c}_{0:t})$. Store both and repeat with $t=t+1$.
\end{enumerate}

\subsection{Computational complexity}
\label{sec:computational_complexity}
The computational complexity of the algorithms in \Cref{sec:gaussian_state_estimation} compares to that of a conventional Gaussian filter/smoother as follows.
Everything in \Cref{sec:model_reduction} happens before encountering realisations, so it shall not count towards the runtime complexity.
Only the procedures in \Cref{sec:gaussian_state_estimation} do.
The following analysis exclusively counts the complexity of matrix-matrix operations (including QR) because, at least asymptotically, these should dominate the algorithm's runtime as the systems increase in size.
Constant multiplicative factors are almost always ignored; we use $\mathcal{O}(\cdot)$ notation.

The computational complexity of the QR decomposition of an $n \times m$ matrix ($n \geq m$) is $\mathcal{O}(n m^2)$. Solving a triangular linear system with $k$ rows and columns via backward substitution costs $\mathcal{O}(k^2)$.
Thus, \Cref{proposition-backward-kernels} with a $k$-dimensional prior and an $m$-dimensional observation costs $\mathcal{O}((k + m)^3 + k m^2)$.

Recall the dimensions of the state variables: $x^\mathsf{u} \in \mathbb{R}^{n-\ell}$, $y^\mathsf{c} \in \mathbb{R}^{\ell}$, and $y^\mathsf{u} \in \mathbb{R}^r$.
\Cref{table:blas_level_iii_operations_after_seeing_data_} lists the floating-point operation count of a single step of the algorithm in \Cref{sec:the_algorithm} in ``filtering mode'', which means that backward densities are not stored (see step 2)). 
\begin{table}[t]
\caption{Matrix-matrix operations in a single step of a numerically robust Kalman filter on the reduced model.}
\label{table:blas_level_iii_operations_after_seeing_data_}
\begin{center}
\begin{tabular}{| c | l | c | c |}
\hline
\bf Step & \bf Operation & \bf Matrix size & \bf Complexity $\mathcal{O}(\cdot)$
\\ 
\hline
1)& QR of a square matrix & $n$ & $n^3$
\\
1)& $n-\ell$ backward subst. & $\ell$  & $(n-\ell) \ell^2$
\\
2) & QR of a square matrix & $2(n-\ell)$  & $(n-\ell)^3$
\\
3) & QR of a square matrix & $n+r-\ell$ & $(n+r-\ell)^3$
\\
3) & $r$ backward subst. & $r$ & $r^3$
\\
\hline 
\end{tabular}
\end{center}
\end{table}
For reference, the operation count of a numerically robust Gaussian filter on the unreduced system is $\mathcal{O}(n^3 + (n+m)^3 + n m^2)$.
The ratio of the total of each set of complexities describes the reduction in runtime when switching from one algorithm to the other.
For example, if $\ell = n/2$ and $r=0$, our filter requires $0.3$ of the floating point operations of the unreduced filter (\Cref{sec:experiments}).
The experiments demonstrate that such ratios are realised for sufficiently large systems.

\section{Experiments}
\label{sec:experiments}
Two experiments are considered.\footnote{An open-source JAX implementation is on GitHub: \\ \url{https://github.com/pnkraemer/code-robust-state-estimation-singular-noise}.} 
The first one demonstrates that the model reduction improves the runtime of a (numerically robust) filter. The second one underlines that our algorithm improves the numerical robustness of a fixed-interval smoother.
\subsection{Demonstrate reduction in computational complexity}
For a set of values for $n$, $\ell$, and $r$, we populate all parameters in \Cref{equation-full-ssm} with samples from independent Gaussian variables with zero mean and unit covariance. For each run, we sample $T=50$ observations from this model.
When evaluating the runtimes of different filters, the precise values of the system parameters are irrelevant; only the matrix sizes matter.
We measure the wall time (fastest of three runs, single precision) of two numerically robust Kalman filters including marginal likelihoods: one that operates on the reduced model and one that operates on the unreduced model. We predict the performance gains via \Cref{table:blas_level_iii_operations_after_seeing_data_} and show the predicted and realised ratios in \Cref{table:ratio_of_runtimes_for_varying_n_m_and_r_reduced_model_over_unreduced_model}.
\begin{table}[t]
\caption{Runtime ratios for varying $n$, $m$, and $r$--- reduced over unreduced version---in a randomly populated state-space model. }
\label{table:ratio_of_runtimes_for_varying_n_m_and_r_reduced_model_over_unreduced_model}
\begin{center}
\begin{tabular}{ | c | c | c | c | c | c | c | }
\hline
$\ell$ & $r$ & $n=10$ & $n=100$ & $n=1000$ & Prediction \\
\hline
$n/2$ & $0$ & 0.81 & 0.44 & 0.28 & 0.30 \\
$n/4$ & $0$ & 1.06 & 0.70 & 0.57 & 0.63 \\
$n/4$ & $n/4$ & 1.32 & 0.92 & 0.57 & 0.54 \\
$n/8$ & $n/8$ & 1.15 & 1.05 & 0.87 & 0.89 \\
\hline
\end{tabular}
\end{center}
\end{table}
Two observations:
First, most ratios are strictly below one, which implies that model reduction speeds up computation. 
Second, as $n$ increases, the runtime ratios approach the predictions. 
In summary, reducing the model improves the numerical efficiency of a (robust) Gaussian filter.

\subsection{Demonstrate improved numerical robustness}
With numerical efficiency covered, next, we demonstrate how our version of a Gaussian smoother is more robust than existing Gaussian smoothers on reduced models \citep{AitElFquih2011}.
To this end, select a range of $n$ and $\ell$, set $\Phi_t = I_n$, $C_t = (I_\ell, 0)$, $r=0$, and choose $Q_t = H_n$, where $H_n$ is a Hilbert matrix with $n$ rows.
The motivation for this model is a combination of simplicity and the fact that the notoriously ill-conditioned Hilbert matrix poses numerical stability challenges for smoothing -- the larger $n$, the worse the conditioning \citep{AitElFquih2011,kramer2024stable}.
We sample $T=500$ observations and evaluate a reference mean and covariance of $p(x_0 \mid y_{0:T})$ with a numerically-robust fixed-point smoother \citep{kramer2025numerically}.
Then, we compute corresponding mean and covariance estimates of three smoothers: (i) our algorithm, (ii) a conventional fixed-interval smoother in a reduced model that uses Cholesky-decompositions to solve linear systems, and (iii) same as (ii), but using LU decompositions to solve linear systems -- lower is better, and a $\log_{10}$-mean-absolute-error (MAE) of negative infinity would be perfect reconstruction.
\Cref{table:log10_mean_absolute_error_of_reconstructing_x_0_mid_y_0_t} shows the $\log_{10}$-MAEs of mean and covariances.
\begin{table}[t]
\caption{Log10-MAEs in the Hilbert-matrix model. Lower is better.}
\label{table:log10_mean_absolute_error_of_reconstructing_x_0_mid_y_0_t}
\begin{center}
\begin{tabular}{| c | c | c | c | c |}
\hline
$n$ & $\ell$ & Ours & Conventional (LU) & Conventional (Chol.) \\
\hline
5 & 2 & -17.7 & -17.5 & -17.7 \\
6 & 3 & -17.6 & -16.7 & -16.5 \\
7 & 3 & -17.7 & 3.0 & NaN \\
8 & 4 & -17.3 & 58.5 & NaN \\
9 & 4 & -15.9 & 207.4 & NaN \\
10 & 5 & -14.5 & 144.5 & 263.7 \\
11 & 5 & -5.7 & NaN & NaN \\
\hline
\end{tabular}
\end{center}
\end{table}
\Cref{table:log10_mean_absolute_error_of_reconstructing_x_0_mid_y_0_t} shows how our algorithm is more robust than the others because the $\log_{10}$-mean-absolute-errors are much lower.

\section{Conclusion}
This article presented a set of numerically robust Gaussian filters and smoothers that, in contrast to existing methods, simultaneously address filtering, smoothing, numerical robustness, and marginal likelihood computation. 
To get there, a sequence of QR-style decompositions was combined with Bayes' rule to first isolate and then eliminate fully determined components of the state. Paired with sequential factorisations, our approach enabled numerically stable conditioning of Gaussian variables, resulting in efficient and robust state estimation algorithms. 
A series of numerical experiments validated these improvements.

\bibliographystyle{apalike}
\bibliography{refs}

\end{document}